\documentclass[a4paper,12pt]{amsart}
\usepackage[total={440pt, 644pt},marginratio={1:1, 4:3}]{geometry}

\title[Generalized Semimagic Squares for Digital Halftoning]
      {Generalized Semimagic Squares\\for Digital Halftoning}
\author{Akitoshi~\textsc{Kawamura}}

\usepackage[scaled]{helvet}
\usepackage{mathrsfs}
\let\mathcal\mathscr

\usepackage{amsmath,amssymb,amsthm}
\newtheorem{theorem}{Theorem}
\newtheorem{lemma}[theorem]{Lemma}

\usepackage[dvipdfm]{graphicx}

\newcommand{\Zset}{\mathbf Z}
\newcommand{\Rset}{\mathbf R}

\begin{document}

\begin{abstract}
  Completing Aronov et~al.'s study on 
  zero-discrepancy matrices for digital halftoning, 
  we determine all $(m, n, k, l)$ for which it is possible to 
  put $m n$ consecutive integers on an $m \times n$ board
  (with wrap-around) 
  so that each $k \times l$ region has the same sum. 
  For one of the cases where this is impossible, 
  we give a heuristic method to find a matrix with small discrepancy. 
\end{abstract}

\maketitle
\thispagestyle{empty}

  A \emph{semimagic square} is a square matrix
  whose entries are consecutive integers 
  and which has equal row and column sums. 
  One way to generalize this millennia-old concept is 
  to specify the sums on regions 
  other than rows and columns. 
  Ingenious constructions of 
  squares satisfying various sum constraints have been described by 
  many professional and amateur mathematicians. 
  While most of them are interested in adding more and more constraints 
  to make their squares impressive, 
  one can generally consider sum conditions on 
  any set of regions. 

  Aronov et al.~\cite{aronov08:_gener_of_magic_squar_with} 
  took up this problem for square regions: 
  is there an $n \times n$ matrix with entries $0$, \ldots, $n ^2 - 1$
  such that every $k \times k$ region has the same sum? 
  It is amusing to note that this variant of the classical problem 
  is motivated by an engineering question of
  finding good dither matrices for \emph{digital halftoning}, 
  a method to approximate a continuous-tone image by
  a binary image for printing (see their paper for details). 
  They showed~%
\cite[Theorem~1]{aronov08:_gener_of_magic_squar_with}, 
  using what they call \emph{constant-gap matrices}, 
  that the answer is yes if $k$ and $n$ are even 
  or if $n$ is an integer power of $k$, 
  and no if $k$ and $n$ are relatively prime or if $k$ is odd and $n$ is even. 
  We will solve this problem completely 
  by determining all $(n, k)$ for which such matrices exist
  (Section~\ref{section: uniform}). 
  Our construction of the matrices is much simpler even for the cases 
  that have already been settled positively. 
  We also give counterexamples to Asano et al.'s conjecture 
  on the smallest possible discrepancy when $n$ is odd and $k = 2$
  (Section~\ref{section: low discrepancy}). 

\subsection*{Definitions}
  For a positive integer~$N$, 
  we write $[N] = \{0, 1, \dots, N - 1\}$. 
  The remainder when an integer~$x$ is divided by $N$ 
  belongs to $[N]$ and 
  is denoted by $x \bmod N$. 

  We consider the slightly generalized setting where 
  the matrices and regions are rectangles instead of squares. 
  Let $m$ and $n$ be positive integers. 
  For an $m \times n$ matrix~$D$ and 
  index $(i, j) \in [m] \times [n]$, 
  we denote the $(i, j)$th entry of $D$ by $D (i, j)$. 
  Any set $R \subseteq [m] \times [n]$ of indices is called a \emph{region}. 
  The sum of the numbers on $R$ is denoted by $
 D (R) = \sum _{(i, j) \in R} D (i, j) 
  $. 
  The \emph{discrepancy} of $D$ with respect to 
  a set~$\mathcal R$ of regions is the difference
  between the maximum and minimum $D (R)$ as
  $R$ varies in $\mathcal R$. 
  When it is zero, $D$ is said to be
\emph{$\mathcal R$-uniform}.

  The translate of $R$ by $(a, b) \in \Zset ^2$ is denoted by
\begin{equation}
 R + (a, b) 
= 
  \bigl\{\, 
   \bigl( (i + a) \bmod m, (j + b) \bmod n \bigr) 
  :
   (i, j) \in R
  \,\bigr\}
\subseteq 
 [m] \times [n]. 
\end{equation}
  The set of all translates of $R$ is denoted by $
 \overline R = \bigl\{\, R + (a, b) : (a, b) \in \Zset ^2 \,\bigr\}
  $. 

\begin{figure}
 \begin{center}
\begin{equation*}\arraycolsep0pt\def\cell{\makebox[1.25em]}
  \begin{array}{|c|c|c|c|c|} \hline
   \cell{14} & \cell{1} & \cell{21} & \cell{0} & \cell{18} \\ \hline
   \cell{16} & \cell{13} & \cell{9} & \cell{22} & \cell{4} \\ \hline
   \cell{5} & \cell{17} & \cell{12} & \cell{7} & \cell{19} \\ \hline
   \cell{20} & \cell{2} & \cell{15} & \cell{11} & \cell{8} \\ \hline
   \cell{6} & \cell{24} & \cell{3} & \cell{23} & \cell{10} \\ \hline
  \end{array}
\end{equation*}
 \end{center}
 \caption{This $5 \times 5$ table~$D$
          has discrepancy $8$ with respect to $\overline{[2] \times [2]}$, 
          because $44 \leq D (R) \leq 52$ for every $2 \times 2$ region~$R$.}
 \label{figure: counterexample}
\end{figure}

  By an $m \times n$ \emph{table} we mean 
  an $m \times n$ matrix in which 
  each element of $[m n]$ 
  appears exactly once. 
  We are interested in 
  tables with small (or zero) discrepancy 
  with respect to $\overline{[k] \times [l]}$, 
  the set of all $k$-by-$l$ rectangles
  (Figure~\ref{figure: counterexample}). 

\section{Zero discrepancy}
 \label{section: uniform}

  The greatest common divisor of positive integers $x$ and $y$
  is denoted by $\gcd (x, y)$. 
  The goal of this section is to show the following: 

\begin{theorem} 
 \label{theorem: zero discrepancy}
  Let $m$, $n$, $k$, $l$ be positive integers with $k < m$ and $l < n$. 
  Let $k' = \gcd (k, m)$ and $l' = \gcd(l, n)$. 
  Then there exists a $\overline{[k] \times [l]}$-uniform 
  $m \times n$ table 
  if and only if 
  $k'$ and $l'$ are greater than $1$ and 
  $k' l' (m n - 1)$ is even. 
\end{theorem}

  This is an immediate consequence of the following 
  Lemmas \ref{lemma: gcd} and \ref{lemma: zero discrepancy}. 

\begin{lemma}
 \label{lemma: gcd}
  A $\overline{[k] \times [l]}$-uniform $m \times n$ matrix
  is $\overline{[\gcd(k, m)] \times [\gcd(l, n)]}$-uniform. 
\end{lemma}

\begin{proof}
  Let $D$ be a $\overline{[k] \times [l]}$-uniform $m \times n$ matrix. 
  We will show that $D$ is $\overline{[k'] \times [l]}$-uniform, 
  where $k' = \gcd (k, m)$. 
  We get the conclusion of the lemma 
  by repeating the same argument with rows and columns switched. 

  For each $(i, j) \in [m] \times [n]$, 
  the regions $[k'] \times [l] + (i, j)$ and $[k'] \times [l] + (i + k, j)$
  have the same sum on $D$, 
  because each of them combined with 
  $[k - k'] \times [l] + (i + k', j)$ makes a $k \times l$ rectangle. 
  Thus for each $(i, j) \in [m] \times [n]$, the rectangles 
\begin{equation}
 \label{equation: regions}
 [k'] \times [l] + (i + q k, j), \qquad q \in [m / k'], 
\end{equation}
  all have the same sum on $D$. 
  Since $k' = \gcd (k, m)$, 
  these $m / k'$ rectangles 
  cover the strip $[m] \times [l] + (0, j)$ without overlap. 
  Hence, 
\begin{align}
  \frac{m}{k'} \cdot D \bigl( [k'] \times [l] + (i, j) \bigr)
&
 =
  \sum _{q \in [m / k']} D \bigl( [k'] \times [l] + (i + q k, j) \bigr)
\\
\notag
&
 =
  D ([m] \times [l] + (0, j))
=
  \frac 1 k \sum _{r \in [m]} D ([k] \times [l] + (r, j)). 
\end{align}
  Since the rightmost side is a constant independent of $(i, j)$ by 
  $\overline{[k] \times [l]}$-uniformity, 
  so is the leftmost side. 
  Thus $D$ is $\overline{[k'] \times [l]}$-uniform. 
\end{proof}

\begin{lemma}
 \label{lemma: zero discrepancy}
  Let $m$ and $n$ be positive integers, and 
  let $k < m$ and $l < n$ be their positive divisors, respectively. 
  Then there exists a $\overline{[k] \times [l]}$-uniform 
  $m \times n$ table
  if and only if 
  $k$ and $l$ are greater than $1$ and 
  $k l (m n - 1)$ is even. 
\end{lemma}

  One direction is a simple
  generalization of
  \cite[Theorem~1 (b, c)]{aronov08:_gener_of_magic_squar_with}: 

\begin{proof}[Proof of the ``only if'' part of Lemma~\ref{lemma: zero discrepancy}]
  Let $D$ be a $\overline{[k] \times [l]}$-uniform 
  $m \times n$ table. 
  It is easy to see that $
D (R) = k l (m n - 1) / 2 
  $ for each $R \in \overline{[k] \times [l]}$. 
  Since $D (R)$ must be an integer, 
  the second claim follows. 
  For the first claim, 
  assume $k = 1$ for contradiction (the case $l = 1$ is similar). 
  Then $D ([1] \times [l]) = D ([1] \times [l] + (0, 1))$ 
  and hence $
 D (0, 0) = D (0, l)
  $, contradicting the assumption that $D$ is a table.
\end{proof}

  For the converse, 
  we use the building blocks 
  provided by the following lemma: 

\begin{figure}
\begin{center}\arraycolsep0pt\def\cell{\makebox[1.16em]}
\begin{equation*}
 \begin{gathered}
  \begin{array}{|c|c|c|c|c|c|c|} \hline
   \cell{0} &
   \cell{1} &
   \cell{2} &
   \cell{3} &
   \cell{4} &
   \cell{5} &
   \cell{6} \\ \hline
   \cell{6} &
   \cell{5} &
   \cell{4} &
   \cell{3} &
   \cell{2} &
   \cell{1} &
   \cell{0} \\ \hline
  \end{array}
  \\
  (k, l, n) = (2, 1, 7)
 \end{gathered}
 \qquad\qquad
 \begin{gathered}
  \begin{array}{|c|c|c|c|c|c|c|} \hline
   \cell{0} &
   \cell{1} &
   \cell{2} &
   \cell{3} &
   \cell{4} &
   \cell{5} &
   \cell{6} \\ \hline
   \cell{3} &
   \cell{4} &
   \cell{5} &
   \cell{6} &
   \cell{0} &
   \cell{1} &
   \cell{2} \\ \hline
   \cell{6} &
   \cell{4} &
   \cell{2} &
   \cell{0} &
   \cell{5} &
   \cell{3} &
   \cell{1} \\ \hline
  \end{array}
  \\
  (k, l, n) = (3, 1, 7)
 \end{gathered}
 \qquad\qquad
 \begin{gathered}
  \begin{array}{|c|c|c|c|c|c|c|c|} \hline
   \cell{0} &
   \cell{4} &
   \cell{1} &
   \cell{5} &
   \cell{2} &
   \cell{6} &
   \cell{3} &
   \cell{7} \\ \hline
   \cell{0} &
   \cell{4} &
   \cell{1} &
   \cell{5} &
   \cell{2} &
   \cell{6} &
   \cell{3} &
   \cell{7} \\ \hline
   \cell{7} &
   \cell{6} &
   \cell{5} &
   \cell{4} &
   \cell{3} &
   \cell{2} &
   \cell{1} &
   \cell{0} \\ \hline
  \end{array}
  \\
  (k, l, n) = (3, 2, 8)
 \end{gathered}
\end{equation*}
 \caption{Examples of matrices of Lemma~\ref{lemma: block Q}.}
 \label{figure: construction of Q}
\end{center}
\end{figure}

\begin{lemma}
 \label{lemma: block Q}
  Let $k > 1$ and $l > 0$ be integers and 
  let $n$ be a positive multiple of $l$. 
  If $k l (n - 1)$ is even, then
  there exists a $
\overline{[k] \times [l]}
  $-uniform $k \times n$ matrix 
  in which each row is a permutation of $[n]$. 
\end{lemma}

\begin{proof}
  A $\overline{[k] \times [l]}$-uniform $k \times n$ matrix and 
  a $\overline{[k'] \times [l]}$-uniform $k' \times n$ matrix 
  stacked vertically make 
  a $\overline{[k + k'] \times [l]}$-uniform $(k + k') \times n$ matrix. 
  Also, a $\overline{[k] \times [l]}$-uniform matrix is 
  $\overline{[k] \times [l']}$-uniform for any multiple~$l'$ of $l$. 
  Therefore, it suffices to construct the desired matrix~$P$ for the cases
  $(k, l) = (2, 1)$, $(3, 1)$ and $(3, 2)$ 
  (Figure~\ref{figure: construction of Q}). 
  If $(k, l) = (2, 1)$, 
  let 
\begin{align}
  P (0, j) 
&
 =
  j, 
&
  P (1, j) 
&
 =
  n - 1 - j. 
\end{align}
  If $(k, l) = (3, 1)$, then $n$ is odd by the assumption; let 
\begin{align}
   P (0, j) 
 &
  = 
   j, 
 &
   P (1, j) 
 &
  = 
   \biggl( j + \frac{n - 1}{2} \biggr) \bmod n, 
 &
   P (2, j) 
 &
  = 
   (-2 j - 1) \bmod n. 
\end{align}
  If $(k, l) = (3, 2)$, let
\begin{align}
  P (0, j) 
 =
  P (1, j) 
&
 =
  \Bigl\lfloor \frac j 2 \Bigr\rfloor + \frac n 2 (j \bmod 2), 
&
 P (2, j) 
&
 =
   n - 1 - j. 
\end{align}
  It is easy to verify that $P$ is $\overline{[k] \times [l]}$-uniform
  in each case. 
\end{proof}

\begin{proof}[Proof of the ``if'' part of Lemma~\ref{lemma: zero discrepancy}]
  We may assume without loss of generality that 
  $l (m n - 1)$ is even. 
  In this case, both $k l (n - 1)$ and $l (m / k - 1)$ are even, 
  so by Lemma~\ref{lemma: block Q}, there are
  a $\overline{[k] \times [l]}$-uniform $k \times n$ matrix $P$
  whose rows are permutations of $[n]$, 
  and 
  a $\overline{[l] \times [1]}$-uniform $l \times (m / k)$ matrix $Q$ 
  whose rows are permutations of $[m / k]$. 
  Define an $m \times l$ matrix~$T$ by 
\begin{equation}
 \label{eq: definition of T}
 T (a, j) = Q (j, \lfloor a / k \rfloor) k + (a \bmod k). 
\end{equation}
  Then $T$ is $\overline{[k] \times [l]}$-uniform and 
  its columns are permutations of $[m]$. 
  Define an $m \times n$ matrix $D$ by
\begin{equation}
 \label{eq: 0608221803}
   D (a, b) 
  =
   P (a \bmod k, b) m + T (a, b \bmod l) 
\end{equation}
  (Figure~\ref{figure: D is P plus T}).
  Since $P$ and $T$ are 
  $\overline{[k] \times [l]}$-uniform, 
  so is $D$. 
  To see that $D$ is a table, 
  suppose that $
   D (a, b) 
  =
   D (a', b')
  $. 
  By \eqref{eq: definition of T} and \eqref{eq: 0608221803} we see that
\begin{equation}
\left\{
\begin{aligned}
 P (a \bmod k, b) & = P (a' \bmod k, b'), \\
 Q (b \bmod l, \lfloor a / k \rfloor) & = Q (b' \bmod l, \lfloor a' / k \rfloor), \\
 a \bmod k & = a' \bmod k. 
\end{aligned}
\right.
\end{equation}
  Since $P$'s rows are permutations, 
  the first and the third equation imply that $b = b'$. 
  Since $Q$'s rows are permutations, 
  this and the second equation imply that $a = a'$. 
\begin{figure}
\begin{center}
\begin{align*}
  D 
&
 = 
   \left[
    \begin{array}{c}
     P \\
     P \\
     P
    \end{array}
   \right]
    \times 9
  + 
   \left[
    \begin{array}{cccc}
     T &
     T & 
     T & 
     T
    \end{array}
   \right]
\\
&
\arraycolsep0pt\def\cell#1#2{\dimen0=1.2em\advance\dimen0\arrayrulewidth\multiply\dimen0 by#1\advance\dimen0-\arrayrulewidth\makebox[\dimen0]{#2}}\footnotesize
 =
  \begin{array}{|c|c|c|c|c|c|c|c|} \hline
   \cell{1}{0} & 
   \cell{1}{4} & 
   \cell{1}{1} & 
   \cell{1}{5} & 
   \cell{1}{2} & 
   \cell{1}{6} & 
   \cell{1}{3} & 
   \cell{1}{7} \\ \hline
   \cell{1}{0} & 
   \cell{1}{4} & 
   \cell{1}{1} & 
   \cell{1}{5} & 
   \cell{1}{2} & 
   \cell{1}{6} & 
   \cell{1}{3} & 
   \cell{1}{7} \\ \hline
   \cell{1}{7} & 
   \cell{1}{6} & 
   \cell{1}{5} & 
   \cell{1}{4} & 
   \cell{1}{3} & 
   \cell{1}{2} & 
   \cell{1}{1} & 
   \cell{1}{0} \\ \hline
   \cell{1}{0} & 
   \cell{1}{4} & 
   \cell{1}{1} & 
   \cell{1}{5} & 
   \cell{1}{2} & 
   \cell{1}{6} & 
   \cell{1}{3} & 
   \cell{1}{7} \\ \hline
   \cell{1}{0} & 
   \cell{1}{4} & 
   \cell{1}{1} & 
   \cell{1}{5} & 
   \cell{1}{2} & 
   \cell{1}{6} & 
   \cell{1}{3} & 
   \cell{1}{7} \\ \hline
   \cell{1}{7} & 
   \cell{1}{6} & 
   \cell{1}{5} & 
   \cell{1}{4} & 
   \cell{1}{3} & 
   \cell{1}{2} & 
   \cell{1}{1} & 
   \cell{1}{0} \\ \hline
   \cell{1}{0} & 
   \cell{1}{4} & 
   \cell{1}{1} & 
   \cell{1}{5} & 
   \cell{1}{2} & 
   \cell{1}{6} & 
   \cell{1}{3} & 
   \cell{1}{7} \\ \hline
   \cell{1}{0} & 
   \cell{1}{4} & 
   \cell{1}{1} & 
   \cell{1}{5} & 
   \cell{1}{2} & 
   \cell{1}{6} & 
   \cell{1}{3} & 
   \cell{1}{7} \\ \hline
   \cell{1}{7} & 
   \cell{1}{6} & 
   \cell{1}{5} & 
   \cell{1}{4} & 
   \cell{1}{3} & 
   \cell{1}{2} & 
   \cell{1}{1} & 
   \cell{1}{0} \\ \hline
  \end{array}
 \times 9 +
  \begin{array}{|c|c|c|c|c|c|c|c|} \hline
   \cell{1}{0} & 
   \cell{1}{6} & 
   \cell{1}{0} & 
   \cell{1}{6} & 
   \cell{1}{0} & 
   \cell{1}{6} & 
   \cell{1}{0} & 
   \cell{1}{6} \\ \hline
   \cell{1}{1} & 
   \cell{1}{7} & 
   \cell{1}{1} & 
   \cell{1}{7} & 
   \cell{1}{1} & 
   \cell{1}{7} & 
   \cell{1}{1} & 
   \cell{1}{7} \\ \hline
   \cell{1}{2} & 
   \cell{1}{8} & 
   \cell{1}{2} & 
   \cell{1}{8} & 
   \cell{1}{2} & 
   \cell{1}{8} & 
   \cell{1}{2} & 
   \cell{1}{8} \\ \hline
   \cell{1}{3} & 
   \cell{1}{3} & 
   \cell{1}{3} & 
   \cell{1}{3} & 
   \cell{1}{3} & 
   \cell{1}{3} & 
   \cell{1}{3} & 
   \cell{1}{3} \\ \hline
   \cell{1}{4} & 
   \cell{1}{4} & 
   \cell{1}{4} & 
   \cell{1}{4} & 
   \cell{1}{4} & 
   \cell{1}{4} & 
   \cell{1}{4} & 
   \cell{1}{4} \\ \hline
   \cell{1}{5} & 
   \cell{1}{5} & 
   \cell{1}{5} & 
   \cell{1}{5} & 
   \cell{1}{5} & 
   \cell{1}{5} & 
   \cell{1}{5} & 
   \cell{1}{5} \\ \hline
   \cell{1}{6} & 
   \cell{1}{0} & 
   \cell{1}{6} & 
   \cell{1}{0} & 
   \cell{1}{6} & 
   \cell{1}{0} & 
   \cell{1}{6} & 
   \cell{1}{0} \\ \hline
   \cell{1}{7} & 
   \cell{1}{1} & 
   \cell{1}{7} & 
   \cell{1}{1} & 
   \cell{1}{7} & 
   \cell{1}{1} & 
   \cell{1}{7} & 
   \cell{1}{1} \\ \hline
   \cell{1}{8} & 
   \cell{1}{2} & 
   \cell{1}{8} & 
   \cell{1}{2} & 
   \cell{1}{8} & 
   \cell{1}{2} & 
   \cell{1}{8} & 
   \cell{1}{2} \\ \hline
  \end{array}
 =
  \begin{array}{|c|c|c|c|c|c|c|c|} \hline
   \cell{1}{0} & 
   \cell{1}{42} & 
   \cell{1}{9} & 
   \cell{1}{51} & 
   \cell{1}{18} & 
   \cell{1}{60} & 
   \cell{1}{27} & 
   \cell{1}{69} \\ \hline
   \cell{1}{1} & 
   \cell{1}{43} & 
   \cell{1}{10} & 
   \cell{1}{52} & 
   \cell{1}{19} & 
   \cell{1}{61} & 
   \cell{1}{28} & 
   \cell{1}{70} \\ \hline
   \cell{1}{65} & 
   \cell{1}{62} & 
   \cell{1}{47} & 
   \cell{1}{44} & 
   \cell{1}{29} & 
   \cell{1}{26} & 
   \cell{1}{11} & 
   \cell{1}{8} \\ \hline
   \cell{1}{3} & 
   \cell{1}{39} & 
   \cell{1}{12} & 
   \cell{1}{48} & 
   \cell{1}{21} & 
   \cell{1}{57} & 
   \cell{1}{30} & 
   \cell{1}{66} \\ \hline
   \cell{1}{4} & 
   \cell{1}{40} & 
   \cell{1}{13} & 
   \cell{1}{49} & 
   \cell{1}{22} & 
   \cell{1}{58} & 
   \cell{1}{31} & 
   \cell{1}{67} \\ \hline
   \cell{1}{68} & 
   \cell{1}{59} & 
   \cell{1}{50} & 
   \cell{1}{41} & 
   \cell{1}{32} & 
   \cell{1}{23} & 
   \cell{1}{14} & 
   \cell{1}{5} \\ \hline
   \cell{1}{6} & 
   \cell{1}{36} & 
   \cell{1}{15} & 
   \cell{1}{45} & 
   \cell{1}{24} & 
   \cell{1}{54} & 
   \cell{1}{33} & 
   \cell{1}{63} \\ \hline
   \cell{1}{7} & 
   \cell{1}{37} & 
   \cell{1}{16} & 
   \cell{1}{46} & 
   \cell{1}{25} & 
   \cell{1}{55} & 
   \cell{1}{34} & 
   \cell{1}{64} \\ \hline
   \cell{1}{71} & 
   \cell{1}{56} & 
   \cell{1}{53} & 
   \cell{1}{38} & 
   \cell{1}{35} & 
   \cell{1}{20} & 
   \cell{1}{17} & 
   \cell{1}{2} \\ \hline
  \end{array}
\end{align*}
\end{center}
 \caption{Construction of $D$ 
          for $(k, l, m, n) = (3, 2, 9, 8)$.}
 \label{figure: D is P plus T}
\end{figure}
\end{proof}

  In the above, we constructed the uniform table as a
  linear combination of two uniform matrices with smaller entries. 
  This idea is due to Euler~%
\cite{euler} 
  who gave a construction of a semimagic square
  (that is, a $(\overline{[1] \times [n]} \cup \overline{[n] \times [1]})$-uniform
  $n \times n$ table) from a pair of special
  $(\overline{[1] \times [n]} \cup \overline{[n] \times [1]})$-uniform matrices 
  called \emph{Latin squares}. 

\section{Finding low-discrepancy tables by ranking}
 \label{section: low discrepancy}

  In this section, we confine ourselves, 
  as Asano et al.~\cite{asano05:_distr_distin_integ_unifor_squar} did, 
  to the case where $k = l = 2$ and $m = n$. 
  Theorem~\ref{theorem: zero discrepancy} states that in this case
  a uniform table exists if and only if $n$ is even. 
  For odd $n$'s, 
  they construct a table with discrepancy $2 n$, 
  and conjecture that it is the smallest possible. 
  This is refuted by our Figures \ref{figure: counterexample}
  and \ref{figure: thirty-one}. 
  Figure~\ref{figure: counterexample} was discovered by an exhaustive search. 
  We describe brief\textcompwordmark ly 
  how Figure~\ref{figure: thirty-one} was obtained. 

\begin{figure}
\begin{center}
\begin{tiny}\tabcolsep1.2pt
\mbox{}\kern-40pt%
\begin{tabular}{rrrrrrrrrrrrrrrrrrrrrrrrrrrrrrr}
433 & 523 & 439 & 519 & 445 & 511 & 453 & 507 & 460 & 497 & 465 & 490 & 472 & 486 & 478 & 480 & 482 & 476 & 487 & 470 & 492 & 462 & 500 & 458 & 508 & 450 & 515 & 442 & 520 & 437 & 525 \\
616 & 348 & 603 & 362 & 592 & 378 & 576 & 393 & 556 & 411 & 543 & 427 & 524 & 449 & 494 & 479 & 471 & 505 & 441 & 529 & 422 & 547 & 403 & 564 & 387 & 581 & 374 & 594 & 357 & 607 & 345 \\
283 & 675 & 298 & 655 & 312 & 634 & 337 & 614 & 356 & 590 & 383 & 560 & 416 & 531 & 451 & 483 & 498 & 436 & 538 & 405 & 570 & 377 & 595 & 350 & 622 & 327 & 643 & 306 & 662 & 292 & 678 \\
727 & 241 & 716 & 253 & 692 & 278 & 669 & 305 & 637 & 340 & 604 & 375 & 569 & 418 & 522 & 475 & 454 & 535 & 401 & 580 & 361 & 615 & 329 & 646 & 300 & 677 & 273 & 701 & 249 & 721 & 238 \\
198 & 762 & 206 & 747 & 225 & 720 & 255 & 686 & 293 & 649 & 328 & 608 & 376 & 559 & 428 & 488 & 517 & 413 & 578 & 358 & 625 & 320 & 660 & 279 & 695 & 248 & 729 & 218 & 751 & 201 & 765 \\
797 & 166 & 786 & 180 & 769 & 207 & 733 & 245 & 698 & 282 & 652 & 330 & 602 & 385 & 539 & 469 & 434 & 558 & 371 & 621 & 315 & 670 & 270 & 706 & 233 & 746 & 199 & 775 & 173 & 790 & 163 \\
132 & 823 & 141 & 813 & 160 & 782 & 196 & 745 & 237 & 703 & 284 & 648 & 342 & 588 & 414 & 491 & 530 & 388 & 606 & 324 & 668 & 269 & 714 & 223 & 755 & 185 & 792 & 156 & 818 & 137 & 827 \\
856 & 109 & 848 & 120 & 826 & 150 & 793 & 191 & 749 & 239 & 696 & 294 & 635 & 359 & 553 & 464 & 425 & 583 & 341 & 653 & 275 & 712 & 221 & 764 & 174 & 802 & 140 & 831 & 116 & 852 & 107 \\
82 & 876 & 88 & 862 & 113 & 833 & 145 & 796 & 192 & 744 & 247 & 684 & 307 & 610 & 397 & 496 & 541 & 370 & 633 & 291 & 704 & 226 & 761 & 172 & 810 & 131 & 843 & 102 & 867 & 84 & 879 \\
898 & 61 & 893 & 78 & 869 & 103 & 836 & 147 & 791 & 197 & 732 & 258 & 667 & 339 & 572 & 461 & 417 & 598 & 310 & 685 & 242 & 752 & 177 & 808 & 130 & 851 & 94 & 880 & 73 & 895 & 58 \\
43 & 916 & 50 & 900 & 71 & 872 & 104 & 832 & 152 & 781 & 210 & 718 & 280 & 632 & 381 & 504 & 550 & 347 & 658 & 259 & 736 & 193 & 798 & 136 & 849 & 91 & 884 & 65 & 906 & 48 & 919 \\
933 & 30 & 924 & 46 & 903 & 72 & 868 & 114 & 824 & 161 & 768 & 229 & 688 & 317 & 586 & 456 & 402 & 623 & 290 & 715 & 209 & 783 & 148 & 839 & 100 & 883 & 59 & 912 & 37 & 928 & 28 \\
14 & 941 & 25 & 927 & 47 & 899 & 79 & 861 & 123 & 811 & 181 & 742 & 256 & 651 & 369 & 509 & 561 & 331 & 682 & 234 & 766 & 162 & 825 & 111 & 871 & 68 & 911 & 36 & 934 & 20 & 943 \\
953 & 10 & 944 & 26 & 923 & 51 & 892 & 89 & 846 & 144 & 784 & 208 & 713 & 299 & 597 & 448 & 394 & 638 & 267 & 734 & 190 & 804 & 126 & 859 & 80 & 901 & 42 & 930 & 18 & 950 & 8 \\
2 & 955 & 11 & 940 & 31 & 915 & 64 & 873 & 112 & 821 & 169 & 757 & 244 & 671 & 354 & 514 & 571 & 314 & 699 & 215 & 778 & 153 & 838 & 95 & 889 & 54 & 921 & 24 & 946 & 6 & 958 \\
960 & 3 & 952 & 17 & 932 & 44 & 896 & 83 & 853 & 134 & 794 & 200 & 725 & 289 & 609 & 444 & 384 & 647 & 257 & 743 & 179 & 814 & 117 & 866 & 70 & 909 & 34 & 939 & 12 & 956 & 0 \\
1 & 957 & 9 & 942 & 29 & 917 & 60 & 875 & 110 & 822 & 167 & 759 & 243 & 673 & 351 & 516 & 574 & 311 & 702 & 214 & 779 & 149 & 840 & 93 & 890 & 52 & 925 & 22 & 949 & 4 & 959 \\
954 & 7 & 947 & 21 & 926 & 49 & 894 & 87 & 850 & 139 & 788 & 205 & 717 & 297 & 599 & 446 & 389 & 642 & 264 & 737 & 186 & 809 & 124 & 863 & 76 & 904 & 39 & 936 & 15 & 951 & 5 \\
13 & 945 & 19 & 931 & 41 & 905 & 75 & 864 & 118 & 815 & 175 & 748 & 251 & 656 & 366 & 510 & 563 & 326 & 687 & 228 & 770 & 158 & 829 & 105 & 878 & 62 & 914 & 32 & 937 & 16 & 948 \\
938 & 27 & 929 & 38 & 910 & 67 & 874 & 106 & 830 & 157 & 774 & 222 & 694 & 309 & 591 & 455 & 400 & 626 & 285 & 722 & 203 & 789 & 142 & 845 & 90 & 887 & 56 & 918 & 33 & 935 & 23 \\
35 & 920 & 45 & 907 & 66 & 882 & 99 & 841 & 146 & 787 & 202 & 726 & 276 & 639 & 379 & 506 & 551 & 344 & 666 & 252 & 741 & 184 & 806 & 128 & 855 & 85 & 891 & 57 & 913 & 40 & 922 \\
908 & 55 & 897 & 69 & 881 & 97 & 847 & 135 & 799 & 189 & 739 & 250 & 672 & 334 & 575 & 459 & 415 & 605 & 304 & 691 & 232 & 760 & 170 & 817 & 121 & 857 & 86 & 886 & 63 & 902 & 53 \\
74 & 885 & 81 & 870 & 101 & 844 & 133 & 805 & 178 & 753 & 236 & 690 & 301 & 618 & 392 & 499 & 544 & 365 & 641 & 281 & 710 & 216 & 772 & 164 & 819 & 122 & 854 & 92 & 877 & 77 & 888 \\
865 & 98 & 858 & 115 & 834 & 138 & 801 & 176 & 756 & 227 & 705 & 287 & 644 & 355 & 557 & 463 & 423 & 587 & 335 & 661 & 265 & 724 & 212 & 773 & 165 & 816 & 129 & 842 & 108 & 860 & 96 \\
119 & 835 & 127 & 820 & 151 & 795 & 183 & 754 & 224 & 711 & 274 & 657 & 333 & 593 & 410 & 493 & 532 & 382 & 613 & 313 & 676 & 260 & 728 & 213 & 771 & 171 & 803 & 143 & 828 & 125 & 837 \\
812 & 155 & 800 & 168 & 780 & 195 & 750 & 231 & 707 & 272 & 665 & 323 & 611 & 380 & 545 & 468 & 432 & 566 & 360 & 628 & 302 & 679 & 261 & 723 & 219 & 758 & 187 & 785 & 159 & 807 & 154 \\
182 & 776 & 194 & 763 & 211 & 731 & 246 & 700 & 277 & 664 & 321 & 620 & 367 & 567 & 424 & 489 & 521 & 406 & 584 & 349 & 630 & 303 & 674 & 268 & 708 & 235 & 740 & 204 & 767 & 188 & 777 \\
738 & 220 & 730 & 240 & 709 & 263 & 681 & 295 & 650 & 325 & 619 & 364 & 577 & 408 & 527 & 474 & 447 & 542 & 395 & 589 & 352 & 627 & 316 & 659 & 286 & 689 & 254 & 719 & 230 & 735 & 217 \\
262 & 693 & 271 & 680 & 296 & 654 & 318 & 629 & 343 & 601 & 372 & 573 & 404 & 536 & 443 & 485 & 503 & 430 & 548 & 396 & 582 & 363 & 612 & 336 & 640 & 308 & 663 & 288 & 683 & 266 & 697 \\
645 & 322 & 631 & 338 & 617 & 353 & 596 & 373 & 579 & 391 & 554 & 419 & 533 & 438 & 502 & 477 & 466 & 513 & 431 & 540 & 409 & 562 & 386 & 585 & 368 & 600 & 346 & 624 & 332 & 636 & 319 \\
390 & 565 & 399 & 552 & 412 & 546 & 421 & 534 & 429 & 526 & 440 & 512 & 457 & 495 & 473 & 481 & 484 & 467 & 501 & 452 & 518 & 435 & 528 & 426 & 537 & 420 & 549 & 407 & 555 & 398 & 568
\end{tabular}%
\kern-40pt\mbox{}%
\end{tiny}
\caption{A $31 \times 31$ table whose discrepancy 
         with respect to $\overline{[2] \times [2]}$ is $27$.}
\label{figure: thirty-one}
\end{center}
\end{figure}

  Define 
  $f \colon [0, 1] ^2 \to \Rset$ by $
f (x, y) = 
   g (x) + g (y)
  $, where 
\begin{align}
   g (x) = 
    \begin{cases}
     1 - (4 x - 1) ^2 & \text{if} \ x \leq 1 / 2,  \\
     -1 + (4 x - 3) ^2 & \text{if} \ x \geq 1 / 2 
    \end{cases}
\end{align}
  (Figure~\ref{figure: function f}). 
\begin{figure}
\begin{center}
\newcommand{\putlabel}[3]{\smash{\makebox[0pt][r]{\raisebox{#2}{#3}\hspace*{#1}}}}%
\rule{0pt}{153pt}
\smash{\raisebox{15pt}{%
\includegraphics{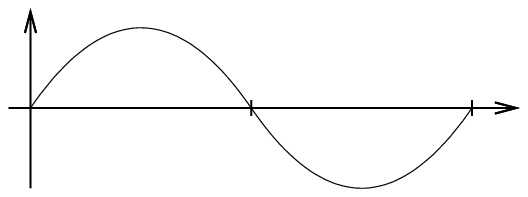}%
\putlabel{-11pt}{30pt}{$x$}%
\putlabel{135pt}{63pt}{$z$}%
\putlabel{6pt}{20pt}{$1$}%
\putlabel{140pt}{20pt}{$0$}%
\putlabel{67pt}{59pt}{$z = g (x)$}%
}}%
\smash{\raisebox{-46pt}{%
\includegraphics{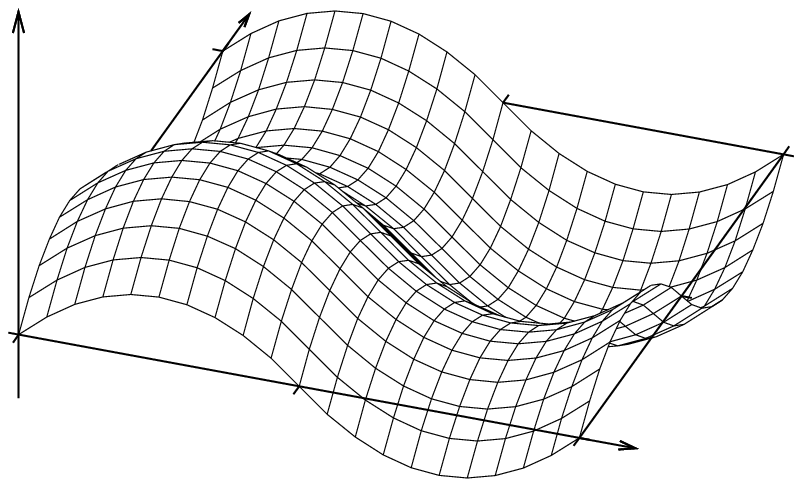}%
\putlabel{87pt}{55pt}{$x$}%
\putlabel{204pt}{188pt}{$y$}%
\putlabel{270pt}{188pt}{$z$}%
\putlabel{280pt}{86pt}{$0$}%
\putlabel{112pt}{48pt}{$1$}%
\putlabel{220pt}{172pt}{$1$}%
\putlabel{105pt}{183pt}{$z = f (x, y)$}%
\hspace{-47pt}\mbox{}%
}}%
 \caption{The functions $g$ and $f$.}
 \label{figure: function f}
\end{center}
\end{figure}
  Let $\alpha$, $\beta \in [0, 1]$ and 
  define $s \colon [n] ^2 \to [0, 1] ^2$ by
\begin{equation}
 \label{eq: grid}
 s (i, j) = \left( \frac{i + \alpha}{n}, \frac{j + \beta}{n} \right). 
\end{equation}
  Let $H$ be the $n \times n$ table whose 
  $(i, j)$th entry is the rank of $f (s (i, j))$
  (with some tie-breaking rule): 
\begin{align}
 \label{eq: ranking}
 H (i, j) = 
  \bigl\lvert 
   \bigl\{\, 
    (i', j') \in [n] ^2 
   :
   \;
& 
     f \bigl( s (i', j') \bigr) < f \bigl( s (i, j) \bigr) 
     \ \text{or} \ 
\\
\notag
&
     \bigl( 
       f \bigl( s (i', j') \bigr) = f \bigl( s (i, j) \bigr) 
      \ \text{and} \ 
       n i' + j' < n i + j
     \bigr)
   \,\bigr\} 
  \bigr\rvert.
\end{align}
  Finally, define the desired matrix~$D$ by 
\begin{equation}
 \label{eq: transformation}
  D \bigl( (i + j) \bmod n, (i - j) \bmod n \bigr) = H (i, j). 
\end{equation}
  Figure~\ref{figure: thirty-one} was obtained by this method 
  with $n = 31$ and $(\alpha, \beta) = (0.286, 0)$. 

  To see intuitively why $D$ has small discrepancy, 
  note that a $2 \times 2$ region in $D$
  corresponds to the region in $H$ (or its translate)
  shown in Figure~\ref{figure: distant domino pair}. 
  These four cells are mapped by $s$ to 
  two nearby points $(x \pm \varepsilon, y)$
  and another two points $(x + 1 / 2, y + 1 / 2 \pm \varepsilon)$ 
  (the coordinates are modulo $1$). 
  Since $f (x, y) = - f (x + 1 / 2, y + 1 / 2)$, 
  the sum of the values of $f$ at these four points is almost zero. 
  Thus, 
  assuming that taking the ranks does not distort the distribution of values too much, 
  we can expect that $D$ has low discrepancy. 
  We add the displacement $(\alpha, \beta)$ in \eqref{eq: grid}
  in order to reduce the chance of ties in the ranking 
  which seem to work adversely. 

\begin{figure}
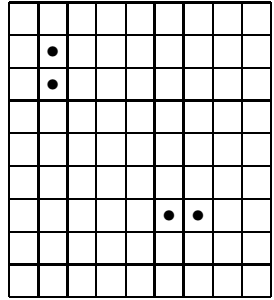

\begin{center}\footnotesize
\newcommand{\colstrut}{\rule{11pt}{0pt}}
\newcommand{\markcell}{$\bullet$}
\tabcolsep0pt
\begin{tabular}{|c|c|c|c|c|c|c|c|c|}
\hline
\colstrut&\colstrut&\colstrut&\colstrut&\colstrut&\colstrut&\colstrut&\colstrut&\colstrut\\ \hline
&\markcell&&&&&&& \\ \hline
&\markcell&&&&&&& \\ \hline
&&&&&&&& \\ \hline
&&&&&&&& \\ \hline
&&&&&&&& \\ \hline
&&&&&\markcell&\markcell&& \\ \hline
&&&&&&&& \\ \hline
&&&&&&&& \\ \hline
\end{tabular}
 \caption{A region in $H$ corresponding to a $2 \times 2$ square in $D$ 
          (for $n = 9$).}
 \label{figure: distant domino pair}
\end{center}
\end{figure}

  As Aronov et al.~%
\cite{aronov08:_gener_of_magic_squar_with}
  point out, 
  our problem is analogous to
  a common situation in discrete geometry 
  where we try to arrange discrete objects so that 
  they look close to some ``balanced'' continuous distribution. 
  The constraint peculiar to our problem is that 
  we have to use each number in $[m n]$ exactly once. 
  The ranking technique used here 
  may be applicable to other problems with this constraint. 
  However, analyzing its performance seems to be hard: 
  although our computer experiment for several $n$'s suggests that 
  the above method achieves sublinear $2 \times 2$ discrepancy, 
  we have no proof yet. 

\subsection*{Acknowledgments}

The author thanks
Tomoko Adachi, Tetsuo Asano, 
Tsukasa Kuribayashi, Yasuko Matsui, Shao-Chin Sung, 
Hideki Tsuiki, Ryuhei Uehara and the referees 
for helpful comments and discussions. 
This work was supported in part by 
the Nakajima Foundation and 
the Natural Sciences and Engineering Research Council of Canada. 
An earlier version was presented at the 
Eighth Japan-Korea Joint Workshop on Algorithms and Computation (WAAC 2005). 

\bibliographystyle{plain}

\end{document}